\documentclass{llncs}

\def \approxversion {}
\ifx \approxversion \undefined
\newcommand{\remove}[1]{#1}
\else
\newcommand{\remove}[1]{}
\fi

\usepackage{amsmath,amssymb,amsfonts,latexsym,url,xspace,algorithmic}
\usepackage{color}
\usepackage{enumerate}
\usepackage{subfigure}
\usepackage[dvips]{graphicx}
\usepackage{url}

\urldef{\mails}\path|{pawasthi, avrim, jamiemmt, osheffet}@cs.cmu.edu|

\newtheorem{fact}{Fact}
\renewcommand{\paragraph}[1]{\noindent\textbf{#1}}

\renewcommand{\P}{\ensuremath{\mathcal{P}}}
\newcommand{\poly}{\ensuremath{\mathrm{poly}}}
\newcommand{\pluck}{{\tt Pluck-a-leaf }}
\newcommand{\splito}{{\tt Split }}
\newcommand{\spliti}{{\tt Split$(i)$ }}
\newcommand{\paste}{{\tt Paste-a-leaf }}
\newcommand{\merge}{{\tt Merge }}

\begin{document}
\mainmatter

\title{Additive Approximation for Near-Perfect Phylogeny Construction 
\thanks{This work was supported in part by the National Science
  Foundation under grant CCF-1116892, by an NSF Graduate Fellowship,
  and by the MSR-CMU Center for Computational Thinking.}}

\author{Pranjal Awasthi \and Avrim Blum \and Jamie Morgenstern \and Or Sheffet }

\institute{Carnegie Mellon University, Pittsburgh,\\ 5000 Forbes Ave., Pittsburgh PA 15213, USA,\\
\mails}

\maketitle

\begin{abstract}
  We study the problem of constructing phylogenetic trees for a given
  set of species. The problem is formulated as that of finding a
  minimum Steiner tree on $n$ points over the Boolean hypercube of
  dimension $d$.  It is known that an optimal tree can be found in
  linear time~\cite{linearding05} if the given dataset has a perfect
  phylogeny, i.e. cost of the optimal phylogeny is exactly
  $d$. Moreover, if the data has a near-perfect phylogeny, i.e. the
  cost of the optimal Steiner tree is $d+q$, it is
  known~\cite{blellochfixed06} that an exact solution can be found in
  running time which is polynomial in the number of species and $d$,
  yet exponential in $q$.  In this work, we give a polynomial-time
  algorithm (in both $d$ and $q$) that finds a phylogenetic tree of
  cost $d+O(q^2)$. This provides the best guarantees known---namely, a
  $(1+o(1))$-approximation---for the case $\log(d) \ll q \ll
  \sqrt{d}$, broadening the range of settings for which near-optimal
  solutions can be efficiently found.  We also discuss the motivation
  and reasoning for studying such additive approximations.
\end{abstract}

\section{Introduction}

Phylogenetics, a subfield of computational biology, aims to construct
simple and accurate descriptions of evolutionary history. These
descriptions are represented as evolutionary trees for a given set of
species, each of which is represented by some set of features
(\cite{gusfield97strings,semple2003phylogenetics}).  A typical choice
for these features are single nucleotide polymorphisms (SNPs), binary
indicator variables for common mutations found in
DNA\cite{Hinds18022005,The2003}; see, for example,
%%, a case which is well-motivated
\cite{blellochfixed06,linearding05,mpalon10,robins300,robinsk00,robinsk05}.
This challenging problem has attracted much attention in recent years,
with progress in studying various computational formulations of this
problem
(\cite{gusfield97strings,multistateblelloch10,blellochfixed06,linearding05,polyfernandez03,mpalon10}). The
problem is often posed as that of constructing the most parsimonious
tree induced by the set of species.

Formally, a \emph{phylogeny} or a \emph{phylogenetic tree} for a set
$C$ of $n$ species, each represented by a string (called taxa) of
length $d$ over a finite alphabet $\Sigma$, is an unrooted tree $T =
(V,E)$ such that $C \subseteq V \subseteq \Sigma^d$. Given a distance
metric $\mu$ over $\Sigma^d$, we define the cost of $T$ as
$\sum_{(u,v)\in E} \mu(u,v)$. The tree of \emph{maximum parsimony} for
a dataset is the tree which minimizes this cost with respect to the
Hamming metric; i.e., it is the optimum Steiner tree for the set $C$
under this metric.

The Steiner tree problem is known to be NP-hard in
general~\cite{karp_reducibility_1972}, and remains NP-hard even in the
case of a binary alphabet with the metric induced by the Hamming
distance~\cite{foulds82}. Extensive recent work, both experimental and
theoretical, has focused on the binary character set with the Hamming
metric
(\cite{gusfield97strings,blellochfixed06,linearding05,polyfernandez03,mpalon10,semple2003phylogenetics,Sridhar:2007:AEN:1322075.1322083,Damaschke:2006:PET:1140638.1140643}).
This version of the phylogeny problem will also be the focus of this
paper.

A phylogeny is called \emph{perfect} if each coordinate $i \in [d]$
flips exactly once in the tree (representing a single mutation of $i$
amongst the set of species)\footnote{Without loss of generality, we
  may assume each coordinate flips at least once, since all
  coordinates on which all species agree may be discarded up
  front.}. If a dataset admits a perfect phylogeny, an optimal tree
can be constructed in polynomial time~\cite{polybaca93} (even linear
time, in the case where the alphabet is
binary~\cite{gusfield97strings}). In this work, we investigate
\emph{near perfect} phylogenies -- instances whose optimal
phylogenetic tree has cost $d+q$, where $q\ll d$. Near perfect
phylogenies have been studied in theoretical
(\cite{multistateblelloch10,blellochfixed06,polyfernandez03,Damaschke:2006:PET:1140638.1140643})
and experimental settings
(\cite{Sridhar:2007:AEN:1322075.1322083}). The work of
\cite{multistateblelloch10,blellochfixed06,polyfernandez03,Damaschke:2006:PET:1140638.1140643}
has given a series of randomized algorithms which find the optimal
phylogeny in running time polynomial in $n$ and $d$ but exponential in
$q$. Clearly, when $q = \omega(\log{d})$, these algorithms are not
tractable.

An alternative approach for finding a phylogenetic tree of low cost is
to use a generic Steiner tree approximation algorithm. The best
current such algorithm yields a tree of cost at most $1.39(d+q)$
\cite{byrkalp10} (we comment that the exponential size of the explicit
hypercube with respect to its small representation size requires one
implement such an algorithm using techniques devised especially for
the hypercube, e.g. Alon et al.~\cite{mpalon10}.) However, notice that
for moderate $q$ (e.g., for $q = polylog(d)$), the {\em excess} of
this tree---meaning the difference between its cost and $d$---may be
extremely large compared to the excess $q$ of the optimal tree.  In
such cases, one would much prefer an algorithm whose excess could be
written as a function of $q$ only.
% For moderate $q$
%(e.g., for $q = o(d)$), such an approximation may be too costly.

In this work, we present a randomized poly$(n,d,q)$-time algorithm
that finds a phylogenetic tree of cost $d + O(q^2)$.

\begin{theorem}
\label{thm:main-theorem}
Given a set $C \subseteq \{0,1\}^d$ of $n$ terminals, such that the
optimal phylogeny of $C$ has cost $d+q$, there exists a randomized
$\poly(n,d,q)$-time algorithm that finds a phylogenetic tree of cost
$d + O(q^2)$ w.p. $\geq 1/2$.
\end{theorem}

Note that Theorem~\ref{thm:main-theorem} provides a substantial
improvement over prior work for the case that $\log{d} \ll q \ll
\sqrt{d}$.  In this range, the exact algorithms are no longer
tractable, and the multiplicative approximations yield significantly
worse bounds. Alternatively, viewed as a multiplicative guarantee, in
this range our tree is within a $1+o(1)$ factor of optimal.  To the
best of our knowledge, this is the first work to give an additive
poly-time approximation to either the phylogeny problem or any
(non-trivial) setting the Steiner tree problem. One immediate
question, which remains open, is whether our results can be improved
to $d + o(q^2)$ or perhaps even to $d+O(q)$.
%One immediate question, which remains open, is whether our results can be extended to the case where $\sqrt{d} \leq q \ll d$.

The rest of the paper is organized as follows. After surveying related
work in Section~\ref{subsec:rw}, we detail notation and preliminaries
in Section~\ref{sec:preliminaries}. The presentation of our algorithm
is partitioned into two parts. In Section~\ref{sec:simple-case}, we
present the algorithm for the case where no pair of coordinates is
identical over all terminals (formal definition there). In
Section~\ref{sec:general-case}, we alter the algorithm for the simple
case, in a nontrivial way, so that the modified algorithm finds a
low-cost phylogeny for any dataset. We conclude in
Section~\ref{sec:open-problems} with a discussion, motivating the
problem of near-perfect phylogeny tree from a different perspective,
and present open problems for future research.
\subsection{Related Work}
\label{subsec:rw}
As mentioned in the introduction, the problem of constructing an
optimal phylogeny is NP-complete even when restricted to binary
alphabets \cite{foulds82}.  Schwartz et
al. \cite{multistateblelloch10} give an algorithm based on an Integer
Linear Programming (ILP) formulation to solve the multi-state problem
optimally, and show experimentally the algorithm is efficient on small
instances.  Perfect phylogenies (datasets which admit a tree in which
any coordinate changes exactly once) have optimal parsimony trees
which can be constructed in linear time in the binary case
\cite{linearding05} and in polynomial time for a fixed alphabet
\cite{polyfernandez03}. Unfortunately, finding the perfect phylogeny
for arbitrary alphabets is NP-hard \cite{bodlaendertwostrikes}.
Recent work \cite{blellochfixed06} gives an algorithm to construct
optimal phylogenetic trees for binary, near-perfect phylogenies (where
only a small number of coordinates mutate more than once in the
optimal tree). However, the running time of the algorithm presented in
their work ~\cite{blellochfixed06} is exponential in the number of
additional mutations.  
%To the authors' knowledge, neither a
%polynomial-time (in the number of additional mutations) algorithm nor
%a hardness result for the near-perfect case is known.

There has also been a lot of work on computing multiplicative
approximations to the Steiner tree problem.  A Minimum Spanning Tree
(MST) over the set of terminals achieves an approximation ratio of 2
and a long line of work has led to the current best bound of 1.39
\cite{takahashifirst80,Berman92,springerlink:10.1007/BFb0023489,Karpinski95newapproximation,zelikovsky96,robins300,Hougardy99,robinsk00,robinsk05,byrkalp10}.
The more recent of these papers use a result due to Borchers and Du
\cite{borcherskratio95} showing that an optimal Steiner tree can be
approximated to arbitrary precision using $k$-restricted Steiner
trees.

Some of these approximations to the Steiner tree problem are not
immediately extendable to the problem of constructing phylogenetic
trees. This is because the size of the vertex set for the phylogeny
problem is exponential in $d$ (there are $2^d$ vertices in the
hypercube). If an algorithm works on an explicit representation of the
graph $G$ defined by the hypercube, then it does not solve the
phylogeny problem in polynomial time. However, the line of work
started by Robins and Zelikovsky \cite{robinsk00,robinsk05} used the
notion of $k$-restricted Steiner trees, which {\em can} be efficiently
implemented on the hypercube. In particular, Alon et
al.~\cite{mpalon10} showed that in finding the optimal $k$-restricted
component for a given set of $k$ terminals, it is sufficient to only
consider topologies with the given $k$ terminals at the leaves. Using
this, they were able to extend that work to achieve a 1.55
approximation ratio for the maximum parsimony problem, and a 16/9
approximation for maximum likelihood. Byrka et al. \cite{byrkalp10}
considered a new LP relaxation to the $k$-restricted Steiner tree
problem and achieved an approximation ratio of 1.39, which can be
combined with the topological argument from Alon et
al.~\cite{mpalon10} to achieve the same ratio for phylogenies.

\section{Notation and Preliminaries}
\label{sec:preliminaries}
Our dataset $C\subseteq \{0,1\}^d$ consists of $n$ terminals over $d$
binary coordinates. A Steiner tree (or phylogeny) over $C$ consists of
a tree $T$ on the hypercube that spans $C$ (plus possibly additional
Steiner nodes), where we label each edge $e$ in $T$ with the index $i
\in \{1,\ldots,d\}$ of the coordinate flipped on edge $e$.  The cost
of such a Steiner tree is the number of edges in the tree. Given a
collection of datasets $\P = \{P_1, P_2, \ldots, P_k\} \subseteq C$ we
define the Steiner forest problem as the problem of finding a minimal
Steiner tree on every $P\in\P$ separately. We refer to such collection
as a partition from now on, even though it may contain a subset of the
original terminal set $C$.

In this work, we consider instances $C$ whose minimum Steiner tree has
cost $d+q$, and think of $q = o(\sqrt{d})$ (otherwise, any
off-the-shelf constant approximation algorithm for the Steiner problem
gives a solution of cost $\leq d+O(q^2)$). We fix $T$ to be some
optimal Steiner tree. By optimality, all leaves in $T$ must be
terminals, whereas the internal nodes of $T$ may be either terminals
or non-terminals (non-terminals are called \emph{Steiner nodes}).  We
define a coordinate $i$ to be \emph{good} if exactly one edge in $T$
is labeled $i$, and \emph{bad} if two or more edges in $T$ are labeled
with $i$. We may assume all $d$ coordinates appear in the tree,
otherwise, some coordinates in $C$ are fixed and so the dimensionality
of the problem is less than $d$. Therefore, at most $q$ coordinates
are bad (each bad coordinate flips at least twice and thus adds a cost
of at least 2 to the tree).

Given a coordinate $i$ of a set of terminals $P$, we define an
\emph{$i$-cut} as the partition $P_0 = \{x\in P:\ x_i = 0\}$ and $P_1
= \{x\in P:\ x_i = 1\}$. We call two coordinates $i\neq j$
\emph{interchangeable} if they define the same cut. We now present the
following basic facts which are easy to verify
(see~\cite{blellochfixed06} for proofs).
\begin{fact}
\label{fact:basic-properties}~~
\begin{enumerate}
\item Let $S$ be a set of interchangeable coordinates. Then all
  coordinates in $S$ appear together in the optimal tree $T$, adjacent
  to one another. That is, in $T$ there are paths s.t. for each path:
  all of its edges are labeled by some $i\in S$, all coordinates in
  $S$ have an edge on the path, and all internal nodes on the paths
  aren't terminals and have degree $2$. On these paths, any reordering
  the $S$-labeled edges yields an equivalent optimal tree.
\item For any two good coordinates, $i\neq j$, one side of the $i$-cut
  is contained within one side of the $j$-cut. Equivalently, there
  exist values $b_j$ such that all terminals on one side of the
  $i$-cut have their $j$th coordinate set to $b_j$.
\item Fix any good coordinate $i$ and let $j$ be a good coordinate
  such that all terminals on one side of the $i$-cut have their $j$
  coordinate set to $b_j$. Then both endpoints of the edge labeled $i$
  have their $j$th coordinate set to $b_j$.

\item A good coordinate $i$ and a bad coordinate $i'$ cannot define
  the same cut.

\end{enumerate}
\end{fact}

It immediately follows from Fact~\ref{fact:basic-properties} that for
a given good coordinate $i$ one can efficiently reconstruct the
endpoints of the edge on which $i$ mutates, except for at most $q$
coordinates. This leads us to the following definition. Given $i$, we
denote $D^i$ as the set of all coordinates that are fixed to a constant value $v_i$ on at least
one side of the $i$-cut~(different coordinates may be fixed on
different sides), and we denote $\mathbf{b}^i$ as the vector of the
corresponding values, i.e. $v_i$'s, of the coordinates in $D^i$. The pair $(D^i,
\mathbf{b}^i)$ is called the \emph{pattern} of coordinate $i$. That
set of terminals that \emph{match the pattern} of $i$ is the set
$P_{\textbf{b}^i} = \{x\in P:\ \forall j\in D^i, ~x_j = {b^i_j}\}$.

\section{A Simple Case: Each Coordinate Determines a Distinct Cut}
\label{sec:simple-case}

To show the main ideas behind our algorithm, we first discuss a
special case in which no two coordinates $i$ and $j$ define the same
cut on the terminal set $C$. Algorithms for constructing phylogenetic
trees often make this assumption as they preprocess $C$ by contracting
any pair of interchangeable coordinates. However, in our case such
contractions are problematic, as we discuss in the next section. So in
Section~\ref{sec:general-case}, when we deal with the general case, we
deal with interchangeable coordinates in a non-trivial fashion.

\subsection{Basic Building Blocks}
\label{subsec:building_blocks}

We now turn to the description of our algorithm. On a high level it is
motivated by the notion of maintaining a proper partition of the
terminals.

\begin{definition}
  Call a partition $\P$ \emph{proper} if the forest produced by
  restricting the optimal tree $T$ to the components $P\in \P$ is
  composed of edge disjoint trees.
\end{definition}

Equivalently, the path in $T$ between two nodes $x$ and $y$ in the
same component $P$ of $\P$ does not pass through any node $x'$ in any
different component $P'$ of $\P$.  Clearly, our initial partition, $\P
= \{C\}$, is proper. Our goal is to maintain a proper partition of the
current terminals while decreasing the dimensionality of the problem
in each step. This is implemented by the two subroutines we now
detail.
%\vspace{-.1in}
\\
\paragraph{Pluck a Leaf and Paste a Leaf.}
The first subroutine works by building the optimal phylogeny bottom-up, finding a good coordinate
$i$ adjacent to a leaf terminal $t$ in the tree, and replacing $t$
with its parent ($t$ with $i$ flipped) in the set of
terminals. Observe that if $i$ is a good coordinate, then this removes
the only occurrence of $i$, leaving all terminals in our new dataset
with a fixed $i$ coordinate, thus reducing the dimensionality of the
problem by $1$.

%\vspace{-.25in}
\begin{figure}[ht]
  \fbox{\parbox{\textwidth}{
  \textbf{\tt Pluck-a-leaf} \\
\textsf{\textbf{input:}} A partition $\P$ of current terminals.\\
  \textsf{\textbf{if}} there exists $P\in\P$ and $x\in P$ s.t. some coordinate $i$ is non-constant on $P$, but only the terminal $x$ has $x_i =0$ (or $x_i=1$), then:
  \vspace{-.1in}
\begin{itemize}
\item Set $P' = P \setminus \{x\} \cup \{\bar x^i\}$, where $\bar x^i$ is identical to $x$ except for flipping $i$.
\item Return $x$ and $\mathcal{P}' = \mathcal{P} \setminus \{P\} \cup \{P'\}$.
\end{itemize}
  \vspace{-.1in}
 \textsf{\textbf{else} fail}.
}}
\end{figure}
%\vspace{-.2in}

The matching subroutine to \pluck is {\tt Paste-a-leaf}: if \pluck
succeeds and returns some $(x, \P')$, and we have found a Steiner
forest for the terminals in $\P'$. Then {\tt Paste-a-leaf} merely
connects $x$ with $\bar x^i$ by an edge labeled $i$, then returns the
resulting forest. (We omit formal description.)

\begin{lemma}
\label{clm:pluck-succeeds}
If $\P$ is a proper partition and \pluck succeeds, then $\P'$ is a
proper partition.
\end{lemma}
\begin{proof}[Sketch] Let $T[P]$ be the subtree in which $x$
  resides. We claim that $x$ is a leaf in $T[P]$, attached by an edge
  labeled $i$ to the rest of the terminals. If this indeed is the
  case, then removing $i$ means removing a leaf-adjacent edge from
  $T[P]$ which clearly leaves all components in the forest
  edge-disjoint.

  Wlog $x$ lies on the $i=0$ side of the cut. If $x$ isn't a leaf,
  then at least two disjoint paths connect $x$ to two other
  terminals. Since $\P$ is proper, both these terminals are in
  $P$. This means $T[P]$ crosses the $i$-cut twice, but then we can
  replace $T[P]$ with an even less costly tree in which $i$ is flipped
  once, by projecting the path between the two occurrences of $i$ onto
  the $i=1$ side.
\end{proof}{\qed}

Observe that lemma~\ref{clm:pluck-succeeds} holds only when the
underlying alphabet of the problem is binary. In particular, for a
non-binary alphabet, such $x$ can be a non-leaf.
\\\\
\paragraph{Split and Merge.} When \pluck can no longer find leaves to pluck, we
switch to the second subroutine, one that works by splitting the set
of terminals into two disjoint sets, based on the value of the $i$-th
coordinate. We would like to split our set of terminals according to
the $i$-cut, and recurse on each side separately. But, in order to
properly reconnect the two subproblems, we need to introduce the two
endpoints of the $i$-labeled edge to their respective sides of the
$i$-cut. Our \splito subroutine deals with one particular case in which
these endpoints are easily identified.

\begin{figure}[ht]
  \fbox{\parbox{\textwidth}{
\textbf{\spliti}\\
\textbf{\textsf{input:}} A partition $\P$ of current terminals, a coordinate $i$ that is not
constant on every component of $\P$.\\
\vspace{-.2in}
\begin{itemize}
\item Find a component $P$ on which $i$ isn't constant. Denote the $i$-cut of $P$ as $(P_0, P_1)$.
\item Find $P_{\mathbf{b}^i}$, the set of terminals that match the pattern of $i$.
\item \textbf{\textsf{if}} exists some $x$ which is the \emph{unique} terminal that matches the pattern of $i$ in one side of the cut (that is, if for some $x$ we have $P_{\mathbf{b}^i}\cap P_0 = \{x\}$ or $P_{\mathbf{b}^i}\cap P_1 = \{x\}$)
\begin{itemize}
\item Flip the $i$-th coordinate of $x$, and let $\bar x^i$ be the resulting node.
\item Add $x$ to its side of the $i$-cut, add $\bar x^i$ to the other side of the cut.
\item Return $x,\bar x^i$ and $\P' = \P \setminus \{P\} \cup \{P_0, P_1\}$.
\end{itemize}
{\sf {\bf else} fail}.
\end{itemize}
}}
\end{figure}

The matching subroutine to \splito is {\tt Merge}: Assume \splito succeeds and returns some $(x,\bar x^i, \P')$, and assume we have found a Steiner forest for the terminals in $\P'$. Then {\tt Merge} merely connects $x$ with $\bar x^i$ by an edge labeled $i$, then returns the resulting forest. (Again, formal description is omitted.)

\begin{lemma}
\label{clm:split_succeeds}
Assume $\P$ is a proper partition. Assume \splito is called on a good
coordinate $i$ s.t. the edge labeled $i$ in $T$ has at least one
endpoint which is a terminal. Then the returned partition $\P'$ is
proper.
\end{lemma}
\begin{proof}[Sketch]
  Since $\P$ is proper, then the induced tree $T[P]$ is the only tree
  in the forest that contains the $i$-labeled edge. The lemma then
  follows from showing that $x$ and $\bar x^i$ are the two endpoints
  of $i$-labeled edge in $T[P]$. This follows from the observation
  that the endpoints of the $i$-labeled edge must both match the
  pattern of $i$. Let $u$ be an endpoint and wlog $u$ belongs to the
  $(i=0)$-side of the cut. On all coordinates that are fixed on the
  $(i=0)$-side, $u$ obviously has the right values. All coordinates
  that are fixed on the $(i=1)$-side can only flip on the
  $(i=0)$-side, but only after traversing $u$, so $u$ has them set to
  the value fixed on the $(i=1)$-side.
\end{proof}{\qed}

\subsection{The Algorithm}
\label{subsec:algorithm}

We can now introduce our algorithm. 

\vspace{-.2in}
\begin{figure}[ht]
  \fbox{\parbox{\textwidth}{ 
\textsf{\textbf{ input:}} A partition $\P$ of current
      terminals. Initially, $\P$ is the singleton set $\P = \{C\}$.
\begin{enumerate}
\item \textsf{\textbf{ if}} \pluck succeeds and returns $(x, \P')$
\begin{itemize}
\item recurse on \P', then \paste $x$ back and return the resulting forest.
\end{itemize}
\item \textsf{\textbf{ else-if}} the number of non-constant coordinates on $\P$ is at least $40q^2$
\begin{itemize}
\item Pick a non-constant coordinate $i$ u.a.r and invoke \spliti.
\item \textsf{\textbf{ if}} \splito succeeds: recurse on $\P'$, then \merge $x$ and $\bar x^i$, and return the resulting forest; \textsf{\textbf{otherwise} fail}.
\end{itemize}
\item \textsf{\textbf {else}}
\begin{itemize}
\item For every $P\in\P$ find its MST, $T(P)$, and return the forest $\{T(P)\}$.
\end{itemize}
\end{enumerate}
}}
\caption{\label{alg:simple-case} Algorithm for the simple case}
\end{figure}
\vspace{-.35in}
\begin{theorem}
\label{thm:correctness-of-alg-simple-case}
With probability $\geq 1/2$, the algorithm in
Figure~\ref{alg:simple-case} returns a tree whose cost is at most
$d+O(q^2)$.
\end{theorem}

% Informally, the proof of Theorem~\ref{thm:correctness-of-alg-simple-case} consists of three parts. First, we show that the {\tt pluck-a-leaf} and {\tt split} steps are exact, in that the coordinates processed at those steps will cost exactly $1$ in the output tree, just as they do in $T^{opt}$.  Next, we show that the algorithm does not abort, with high probability. Finally, we show that the coordinates processed in the base of recursion will each cost $O(1)$ in the output tree (since there are $O(q^2)$ coordinates processed in the base step, this implies the theorem).

In order to prove Theorem~\ref{thm:correctness-of-alg-simple-case}, fix an
optimal phylogeny $T$ over our initial set of terminals, and for any partition $\P$ our algorithm creates, denote $T[\P]$ as the forest induced by $T$ on this partition. The proof of the theorem relies on the following lemma.

\begin{lemma}
\label{clm:correctness-step-2}
If $\P$ is a proper partition, then with probability $\geq 1-(8q)^{-1}$,
\splito is called on a good coordinate and succeeds. Furthermore, \splito is executed at most $4q$ times.
\end{lemma}

\begin{proof}[of Theorem~\ref{thm:correctness-of-alg-simple-case}]
  The proof follows from lemmas~\ref{clm:pluck-succeeds} and~\ref{clm:correctness-step-2}. Since we start with a proper partition, then with
probability at least $1 - (4q)(8q)^{-1} \geq 1/2$ we keep recursing on proper partitions, until reaching the base of the recursion. By the time the algorithm reaches the base of the recursion, the dimensionality of the problem was reduced to $d' \leq 40q^2$, so the cost of the optimal Steiner forest is at most $d'+q$. As MSTs give a $2$-approximation to the optimal Steiner tree problem, our forest is of cost $\leq 2(d'+q)$. Then, the algorithm reconnects the forest, adding the coordinates (edges) the algorithm as removed in the first two steps of the algorithm. Since the algorithm removed at most $d-d'$ edges, the tree it outputs is of overall cost at most $d - d'+ 2(d' +q) = d + 40q^2 +  2q$.
\end{proof}{\qed}

\begin{proof}[of Lemma~\ref{clm:correctness-step-2}]
Let $\P$ be the partition in the first iteration of the algorithm for which \splito was invoked, and assume $\P$ is proper. Thus, the forest $T[\P]$ contains disjoint components. We call any vertex in
  this forest of degree $\geq 3$ an \emph{internal split}. Suppose we
  replace each internal split $v$ with $deg(v)$ many new vertices, each adjacent to one edge. This breaks the forest into a collections of paths we call the \emph{path decomposition} of the tree. In addition, remove from this path decomposition all edges that are labeled with a bad coordinate to obtain the \emph{good path decomposition}. Denote the number of paths in the good path decomposition as $t$.

First, we claim that any call to \splito (on $\P$ or any partition succeeding $\P$), on a coordinate $i$ which lies on a path of length $\geq 2$ in the abovementioned decomposition, does not fail.

Assume \splito was called on $i$ and denote its adjacent coordinate on
the path as $j$~(choose one arbitrarily if $i$ has two adjacent
coordinates on its path), and both are non-constant on $P\in
\P$. Observe that our decomposition leaves only good coordinates, so
both $i$ and $j$ are good. Therefore, $j$ is fixed on one side of the
$i$-cut and $i$ is fixed on one side of the $j$-cut. It follows that
there exist binary values $b_i, b_j$ s.t. for every $x\in P$, if
$x_i=b_i$ then $x_j = b_j$; and if $x_j = 1-b_j$ then $x_i= 1-b_i$. In
fact, the only node on the entire tree for which $x_i=1-b_i$ and $x_j
= b_j$ is the node connecting the $i$-edge and the $j$-edge. Recall
that we assume for the special case $i$ and $j$ do not define the same
cut. It follows that the node between $i$ and $j$ has to be a
terminal, so now we can use Lemma~\ref{clm:split_succeeds} and deduce
\splito succeeds.

So, \splito can either fail or return a non-proper partition only if it
was invoked either on a bad coordinate or on a good coordinate that
lies on a path of length $1$ in our path decomposition. There are at
most $q$ bad coordinates and at most $t$ paths of length $1$, so each
call to \splito fails w.p. $\leq \frac {q + t} {40q^2}$. Furthermore,
calling \splito on a good edge $i$ lying on a path of length at least
$2$ results in both $i$'s endpoints as new leaves in their respective
sides of the $i$-cut. As a result, \pluck then completely unravels the
path on which $i$ lies. Therefore, in a successful run of the
algorithm, \splito is called no more than $t$ times. All that remains
is to bound $t$.

$t$ is the number of paths on the path decomposition of $\P$, a
partition for which \pluck failed to execute. Observe that if the
forest $T[\P]$ had even a single leaf connected to the rest of its
tree by a good coordinate, then \pluck would continue -- such a leaf,
by definition, is the only terminal on which the good coordinate
takes a certain value. It follows that $l$, the number of leaves in
$T[\P]$ is bounded by $2q$, the number of bad edges in $T$. Removing
the internal splits then leaves us with at most $2l$ paths; removing
the bad coordinates' edges adds at most $2q-l$ new paths (for every
bad coordinate $k$ adjacent to a leaf, removing $k$ does not create a
new path). All in all, $t \leq 2l + 2q-l \leq 4q$. Therefore, each
call to \splito has success probability $\geq 1 - \frac{4q+q} {40q^2} =
1-\frac 1 {8q}$, and \splito is called at most $4q$ times.
\end{proof}{\qed}

\section{The General Case: Interchangeable Coordinates May Exist}
\label{sec:general-case}

Before describing the general case, let us briefly discuss why the
conventional way of initially contracting all interchangeable
coordinates and applying the algorithm from the
Section~\ref{sec:simple-case} might result in a tree of cost $d +
\omega(q^2)$. The analysis of the first two steps of the algorithm
still holds. The problem lies in the base of the recursion, where the
algorithm runs the MST-based $2$-approximation. Indeed, the MST
algorithm is invoked on $<40q^2$ contracted coordinates, but they
correspond to $\tilde d$ original coordinates, and it is possible that
$\tilde d\gg q^2$. So by using any constant approximation on this
entire forest, we may end with a tree of cost $d+ 2\tilde d$ which
isn't $d+O(q^2)$.

Our revised algorithm does not contract edges initially. Instead, let
us define a \emph{simple} coordinate as one for which \spliti
succeeds. So, the first alteration we make to the algorithm is to call
\splito as long as the set of simple coordinates is sufficiently
big. However, most alterations lie in the base of the recursion. Below
we detail the algorithm and analyze its correctness. In the
algorithm's description, for any coordinate $i$ we denote the set of
coordinates interchangeable with $i$ by $W_i$, and their number as
$w(i)=|W(i)|$.

\begin{theorem}
\label{thm:correctness-of-alg-general-case}
With probability $\geq 1/2$, the algorithm in Figure~\ref{alg:general-case} returns a tree of cost
$d+O(q^2)$.
\end{theorem}

The proof of Theorem~\ref{thm:correctness-of-alg-general-case} follows
the same outline as the proof of
Theorem~\ref{thm:correctness-of-alg-simple-case}. Observe that
Lemmas~\ref{clm:pluck-succeeds} and~\ref{clm:correctness-step-2}
still hold\footnote{Clearly, \splito cannot abort now, but it might be
  the case that the algorithm picks $i$ which is a bad
  coordinate. This can happen with probability $\leq q/8q^2 =
  1/8q$.}. Therefore, with probability $\geq 1/2$, the algorithm
enters the base of the recursion with a proper partition. Thus, by the following lemma, the algorithm outputs a tree of cost $d+O(q^2)$.

\begin{lemma}
\label{clm:correctness-bottom-of-recursion}
Assume that the base of the recursion~(i.e., Step 3) is called on a proper partition $\P$ of the terminals over $d'$ non-constant coordinates. Then
the algorithm returns a forest of cost $d' + O(q^2)$.
\end{lemma}

The full proof of Lemma~\ref{clm:correctness-bottom-of-recursion} is deferred to the appendix. However, let us sketch the main outline of the proof. Recall the good path decomposition we used in the proof of Lemma~\ref{clm:correctness-step-2}. We partition its paths in the following way.

\begin{figure}[ht]
\fbox{\parbox{\textwidth}{
\textsf{\textbf{ input:}} A partition $\P$ of current
      terminals. Initially, $\P$ is the singleton set $\P = \{C\}$.
\begin{enumerate}
\item \textsf{\textbf{ if}} \pluck succeeds and returns $(x, \P')$
\begin{itemize}
\item recurse on \P', then \paste $x$ and return the resulting forest.
\end{itemize}
\item \textsf{\textbf{ else-if}} the number of simple coordinates on $\P$ is at least $8q^2$
\begin{itemize}
\item Pick a simple coordinate $i$ u.a.r and invoke \spliti.
\item \textsf{\textbf{ if}} \splito succeeds: recurse on $\P'$, then \merge $x$ and $\bar x^i$, and return the resulting forest; \textsf {\textbf {otherwise} fail}.
\end{itemize}
\item \textsf{\textbf{ else}}
\begin{itemize}
\item Contract all $W_i$ into $\bar i$
\item For every $\bar i$ with $w(i)>q$ and the (unique) component $P$ in
  which the $i$-cut resides,
\begin{itemize}
\item Apply pattern matching to $(P, i)$. Let $(D^i, \textbf{b}^i)$ be the pattern of $i$. 
\item \textsf{\textbf{if}} $i$ is simple, split $P$ into $P_0\cup\{x^i\}$ and
  $P_1\cup\{\overline{x}^i\}$.
\item  \textsf{\textbf{else}}
\begin{itemize}
\item Define the node $y(i)$ as the node where $y_i = 0$, every
  coordinate $j\in D^i$ is set to $b^i_j$, and every coordinate $j\notin
  D^i$ is set to $0$.
\item Define $\overline {y(i)}$ to be $y(i)$ with coordinate $i$ flipped.
\item  $\P = \P\setminus \{P\} \cup \{P_0 \cup \{y(i)\}\} \cup \{P_1\cup \{\overline{y(i)}\}\}$.
\end{itemize}
\end{itemize}
\item For every $P\in\P$ find its MST $T(P)$, and retrieve
  the forest $\{T(P)\}$.
\item For every $i$ with $w(i)>q$: 
\\\textsf{\textbf{if}} $i$ was simple, add an edge labeled $i$ between $x^i$ and $\overline{x^i}$
\\ \textsf{\textbf{else}} add an edge labeled $i$ between $y(i)$ and  $\overline{y(i)}$.
\item Expand all contracted coordinates to their original set of
  coordinates by replacing $i$ with a path of length $w(i)$. Return the resulting forest.
\end{itemize}
\end{enumerate}
}}
\caption{\label{alg:general-case} Algorithm for the general case}
\vspace{-.3in}
\end{figure}

\begin{itemize}
\item {\it Paths with at least one terminal on them.} On such paths,
  because all interchangeable coordinates may appear in $T$ in any order, then
  all coordinates on such paths are simple. So, when we enter the base of the
  recursion, there are at most $8q^2$ edges on such paths.
\item {\it Paths with no terminal on them, with length $> q$.} Such
  paths are composed of interchangeable coordinates, and since there
  are more than $q$ of those, we deduce all of them are
  good. Therefore, the endpoints of such paths are fixed up to at most
  $q$ (bad) coordinates. We therefore contract these edges, split on
  them, and introduce into each side of the cut an arbitrary endpoint,
  by replacing non-fixed coordinates with zeros. So on each side of
  the cut the cost of the subtree increases by at most $q$, and since
  there are at most $4q$ such paths, our overall cost for introducing these artificial endpoints is $O(q^2)$.
\item {\it Paths with no terminal on them, with length $\leq q$.} Such
  paths are composed of interchangeable coordinates, but we do not
  contract them. Since there are at most $4q\cdot q$ edges on such
  paths, we run the MST approximation, and incur a cost of $O(q^2)$
  for edges on such paths.
\end{itemize}

\noindent\textbf{Runtime analysis:} \pluck can be implemented
in time linear in the size of the dataset, i.e. $O(nd)$. Counting the
number of simple coordinates takes time $O(nd^2)$, and \splito takes time $O(nd)$. A naive implementation of the base case of the recursion
takes time $O(nd^2)$ for contracting coordinates, and the rest can be
implemented in time $O(nd)$. Hence the time to process each node in
the recursion tree is at most $O(nd^2)$. Since there are at most $O(q)$
nodes in the recursion tree, the total runtime is $O(qnd^2)$.

\section{Discussion and Open Problems}
\label{sec:open-problems}

This paper presents a randomized approximation algorithm for
constructing near-perfect phylogenies. In order to achieve this, we
obtain a Steiner tree of low additive error.  However, from the
biological perspective, the goal is to find a good evolutionary tree,
one that will give correct answers to questions like ``what is the
common ancestor of the following species?'' or ``which of the two
gene-mutations happened earlier?''. Such questions, we hope, can be
answered by finding the most-parsimonious phylogenetic tree over the
given taxa. Hence, it is also desirable that any low-cost tree which
we output also captures a lot of the structure of the optimal tree.

We would like to point out that our algorithm in fact has this
valuable property.  Notice that until the base case of the recursion,
both \pluck and \splito subroutines construct
the optimal tree, and correctly identify the endpoints of the edges they remove. Even when the algorithm reaches the base case of the recursion -- we can declare every edge of weight $>q$ to be good, and we know its endpoints up to at most $q$ coordinates. In total, our
algorithm gives the structure of the optimal tree up to $O(q^2)$
edges, and those edges can be marked as ``unsure''.

Several open problems remain for this work. The most straight-forward one is whether one can devise an algorithm outputting a phylogenetic tree of cost $d + O(q)$? Alternatively, one may try to design exact algorithms that are efficient even for $q = \omega(\log d)$. We suspect that even the case of $q = O((\log d)^2)$ poses quite a challenge. Finally, extending our results to non-binary alphabets is intriguing. Note however that even the case of perfect phylogenies is NP-hard, and tractable only for moderately sized alphabets. Furthermore, our bottom-up approach completely breaks down for non-binary alphabets (see comment past Lemma~\ref{clm:pluck-succeeds}), so devising an additive-approximation algorithm for the phylogeny problem with non-binary alphabets requires a different approach altogether.

\bibliographystyle{splncs}
\bibliography{sources}

\begin{thebibliography}{10}

\bibitem{linearding05}
Ding, Z., Filkov, V., Gusfield, D.:
\newblock A linear-time algorithm for the perfect phylogeny haplotyping (pph)
  problem.
\newblock In: RECOMB, Springer (2005)  585--600

\bibitem{blellochfixed06}
Blelloch, G.E., Dhamdhere, K., Halperin, E., Ravi, R., Sridhar, S.:
\newblock Fixed parameter tractability of binary near-perfect phylogenetic tree
  reconstruction.
\newblock In: ICALP, Springer (2006)  667--678

\bibitem{gusfield97strings}
Gusfield, D.:
\newblock Algorithms on strings, trees, and sequences: computer science and
  computational biology.
\newblock Cambridge University Press (1997)

\bibitem{semple2003phylogenetics}
Semple, C., Steel, M.:
\newblock Phylogenetics.
\newblock Oxford lecture series in mathematics and its applications. Oxford
  University Press (2003)

\bibitem{Hinds18022005}
Hinds, D.A., Stuve, L.L., Nilsen, G.B., Halperin, E., Eskin, E., Ballinger,
  D.G., Frazer, K.A., Cox, D.R.:
\newblock Whole-genome patterns of common dna variation in three human
  populations.
\newblock Science \textbf{307}(5712) (2005)  1072--1079

\bibitem{The2003}
:
\newblock The international hapmap project.
\newblock Nature \textbf{426}(6968) (2003)  789--96

\bibitem{mpalon10}
Alon, N., Chor, B., Pardi, F., Rapoport, A.:
\newblock Approximate maximum parsimony and ancestral maximum likelihood.
\newblock IEEE/ACM Trans. Comput. Biol. Bioinformatics \textbf{7} (Jan 2010)
  183--187

\bibitem{robins300}
Robins, G., Zelikovsky, A.:
\newblock Improved steiner tree approximation in graphs.
\newblock In: SODA, Society for Industrial and Applied Mathematics (2000)
  770--779

\bibitem{robinsk00}
Robins, G., Zelikovsky, A.:
\newblock Improved steiner tree approximation in graphs (2000)

\bibitem{robinsk05}
Robins, G., Zelikovsky, A.:
\newblock Tighter bounds for graph steiner tree approximation.
\newblock SIAM Journal on Discrete Mathematics \textbf{19} (2005)  122--134

\bibitem{multistateblelloch10}
Misra, N., Blelloch, G.E., Ravi, R., Schwartz, R.:
\newblock Generalized buneman pruning for inferring the most parsimonious
  multi-state phylogeny.
\newblock In Berger, B., ed.: RECOMB. Volume 6044 of Lecture Notes in Computer
  Science., Springer (Apr. 2010)  369--383

\bibitem{polyfernandez03}
Fern\'{a}ndez-Baca, D., Lagergren, J.:
\newblock A polynomial-time algorithm for near-perfect phylogeny.
\newblock SIAM J. Comput. \textbf{32} (May 2003)  1115--1127

\bibitem{karp_reducibility_1972}
Karp, R.M.:
\newblock Reducibility among combinatorial problems.
\newblock In Miller, R.E., Thatcher, J.W., eds.: Complexity of Computer
  Computations.
\newblock Plenum, New York (1972)  85--103

\bibitem{foulds82}
Foulds, L.R., Graham, R.L.:
\newblock The {S}teiner problem in phylogeny is {NP}-complete.
\newblock Adv. Appl. Math. 3 (1982)

\bibitem{Sridhar:2007:AEN:1322075.1322083}
Sridhar, S., Dhamdhere, K., Blelloch, G., Halperin, E., Ravi, R., Schwartz, R.:
\newblock Algorithms for efficient near-perfect phylogenetic tree
  reconstruction in theory and practice.
\newblock IEEE/ACM Trans. Comput. Biol. Bioinformatics \textbf{4} (Oct 2007)
  561--571

\bibitem{Damaschke:2006:PET:1140638.1140643}
Damaschke, P.:
\newblock Parameterized enumeration, transversals, and imperfect phylogeny
  reconstruction.
\newblock Theor. Comput. Sci. \textbf{351} (Feb 2006)  337--350

\bibitem{polybaca93}
Agarwala, R., Fernandez-Baca, D.:
\newblock A polynomial-time algorithm for the perfect phylogeny problem when
  the number of character states is fixed.
\newblock In: SFCS. (Nov 1993)  140--147

\bibitem{byrkalp10}
Byrka, J., Grandoni, F., Rothvo\ss, T., Sanit\`{a}, L.:
\newblock An improved lp-based approximation for steiner tree.
\newblock In: STOC, ACM (2010)

\bibitem{bodlaendertwostrikes}
Bodlaender, H.L., Fellows, M.R., Warnow, T.:
\newblock Two strikes against perfect phylogeny.
\newblock In: ICALP'92. (1992)  273--283

\bibitem{takahashifirst80}
Takahashi, H., Matsuyama, A.:
\newblock An approximate solution for the steiner problem in graphs.
\newblock Mathematica Japonica \textbf{24} (1980)  573--577

\bibitem{Berman92}
Berman, P., Ramaiyer, V.:
\newblock Improved approximations for the steiner tree problem.
\newblock In: SODA. (1992)  325--334

\bibitem{springerlink:10.1007/BFb0023489}
Pramel, H., Steger, A.:
\newblock Rnc-approximation algorithms for the steiner problem.
\newblock In Reischuk, R., Morvan, M., eds.: STACS 97. Volume 1200 of Lecture
  Notes in Computer Science.
\newblock Springer Berlin / Heidelberg (1997)  559--570 10.1007/BFb0023489.

\bibitem{Karpinski95newapproximation}
Karpinski, M., Zelikovsky, A.:
\newblock New approximation algorithms for the steiner tree problems.
\newblock Journal of Combinatorial Optimization \textbf{1} (1995)  47--65

\bibitem{zelikovsky96}
Zelikovsky, A.:
\newblock Better approximation bounds for the network and euclidean steiner
  tree problems.
\newblock Technical report (1996)

\bibitem{Hougardy99}
Hougardy, S., Promel, H.J.:
\newblock A 1.598 approximation algorithm for the steiner problem in graphs.
\newblock In: SODA. (1999)  448--453

\bibitem{borcherskratio95}
Borchers, A., Du, D.Z.:
\newblock The k-steiner ratio in graphs.
\newblock In: STOC, ACM (1995)  641--649

\end{thebibliography}

\appendix
\section{Appendix}
Here, we give short proofs of the basic properties presented in Fact
~\ref{fact:basic-properties}. Many of these results are also found in
other work ~\cite{blellochfixed06}.

\begin{enumerate}
\item Let $S$ be a set of coordinates that define the same cut over
  the terminals (up to renaming of $P_0$ and $P_1$).  Then whenever
  any edge in $T$ is labeled with some $i\in S$, then the edge lies on
  a path on which each edge is labeled with a unique coordinate from
  $S$, and no node of the $|S|$-long path is a terminal.
\begin{proof}
  Assume $i, j\in S$ define the same cut over $C$. Given some edge
  $e_i$ where $i$ flips, consider the shortest path from a terminal
  $t_1$ through $e_i$ to another terminal $t_2$. If every node on this
  path is of degree 2, there is no terminal before $t_2$, and $j$
  flips before $t_2$.  Now, suppose some node on the path between
  $t_1$ and $t_2$ has degree more than 2. Either $j$ flips on each
  outgoing path before any terminals, or $j$ does not define the same
  cut (since each outgoing path has a terminal on it). But if $j$
  flips on all outgoing paths before any terminal occurs on those
  paths, relabel the Steiner nodes on each path so that $j$ is
  constant along those outgoing paths. Then, add one Steiner node and
  an edge from the endpoint of $e_i$ which flips $j$. This tree has
  cost strictly less than the tree which flipped $j$ on each outgoing
  path, since there were at least two paths, each of which flipped
  $j$.

\end{proof}{\qed}
\item For any two good coordinates, $i\neq j$, one side of the $i$-cut
  is contained within one side of the $j$-cut. Alternatively, there
  exist values $b_j$ such that all terminals on one side of the
  $i$-cut have their $j$th coordinate set to $b_j$.
\begin{proof}
  Suppose $i$ is good. Then, consider the $i$-cut in $T$. Since $j$ is
  good, $j$ may flip only once in the tree. If $j$ flips in $C_0$,
  then $j$ is constant in $C_1$. If $j$ flips in $C_1$, then $j$ is
  constant in $C_0$.
\end{proof}{\qed}
\item Fix any good coordinate $i$ and let $j$ be a good coordinate
  such that all terminals on one side of the $i$-cut have their $j$
  coordinate set to $b_j$. Then both endpoints of the edge labeled $i$
  have their $j^{\textrm{th}}$ coordinate set to $b_j$.
\begin{proof}
  Suppose that some endpoint of the edge on which $i$ flips has $j$
  (which is constant on $C_0$) set to $\bar b_j$. Then, since the edge
  allows only one coordinate to flip across it, both endpoints are
  labelled with $j =\bar b_j$. Then, the side where $j$ is constant
  (say $C_0$) has to pay for $j$.

  If $j$ is constant and set to $b_j$ on $C_1$, $j$ has to flip twice,
  a contradiction since $j$ is good. If $j$ is constant and set to
  $\bar b_j$ on $C_0$, then the labels on $i$'s edge are labelled by
  some constant setting of $j$. If $j$ is non-constant Assuming $j$ is
  non-constant for $C_1$, $j$ flips somewhere in $C_1$. But then, $j$
  flips twice, contradicting the fact that $j$ is good.
\end{proof}{\qed}
\item A good coordinate $i$ and a bad coordinate $i'$ cannot define
  the same cut.
\begin{proof}
  This follows directly from Fact 2.1, since $i$ and $i'$ will occur
  the same number of times in an optimal tree and a good coordinate
  occurs exactly once, while a bad coordinate occurs at least twice.
\end{proof}{\qed}
\end{enumerate}

\section{Proof of Lemma~\ref{clm:correctness-bottom-of-recursion}}
\label{apx_sec:proof_of_claim}

Here we give the full proof of Lemma~\ref{clm:correctness-bottom-of-recursion}. For convenience, we reiterate the lemma.

\begin{lemma}
\label{apx_clm:correctness-bottom-of-recursion}
Assume that when step 3 is entered, $\P$ is a proper
partition of the terminals over $d'$ non-constant coordinates. Then
step 3 of the algorithm returns a forest of cost $d' + O(q^2)$.
\end{lemma}

\begin{proof}
  Recall that for any coordinate $i$, we define the weight of $i$ as the total number of coordinates interchangeable with $i$.
  Let $\P^{\rm post}$ denote the instance which results after
  splitting on all coordinates with weight greater than $q$. The proof
  of the lemma reduces to showing that there exists a forest over
  $\P^{\rm post}$ of cost $O(q^2)$. If such a forest exists,
  the MST-based $2$-approximation for each component returns a forest
  of cost $O(q^2)$, and reconstructing the tree with plucked and split
  edges has weight at most $d'$. So, the forest over $\P$ we
  get has cost $\leq d'+O(q^2)$.

  The set of all terminals in all components of $\P^{\rm
    post}$ is composed of terminals that belong to the original
  $\P$, and the new terminals, added by coordinates of weight
  more than $q$ which are \emph{not} simple. We construct a forest of
  low cost over $\P^{\rm post}$ by (i) taking an optimal
  Steiner forest $\mathcal{T}$ over $\P$, (ii) for every
  simple coordinate which is split upon, remove the path of
  corresponding coordinates from $\mathcal{T}$, and (iii) for every
  non-simple coordinate split upon, remove the path of corresponding
  coordinates \emph{and} introduce the two vertices $y(i)$ and
  $\overline{y(i)}$, connecting each vertex to the end-point of the
  path that resides in the same side of the cut. We denote the
  resulting forest by $\mathcal{T}^{\rm post}$, and show that
  $\mathcal{T}^{\rm post}$ contains at most $O(q^2)$ edges.

  First, observe that since there are at most $q$ bad coordinates, by
  splitting only on coordinates with weight more than $q$, we are
  guaranteed to split only on good coordinates (we know from Fact
  \ref{fact:basic-properties} that the good and bad coordinates cannot
  define the same cut). Therefore, since $\P$ is proper, the
  coordinate $i$ resides in a single component. Second, observe that
  the base of the recursion is executed only when \pluck fails. Therefore, as in the proof of Lemma~\ref{clm:correctness-step-2}, $\mathcal{T}$ is a
  forest with at most $2q$ leaves, so its path decomposition contains
  at most $4q$ paths. Our next observation is the following
  proposition.

\begin{proposition}
\label{pro:terminal-on-a-good-path}
Fix a path in the path decomposition of the forest. If there exists
some non-endpoint terminal on this path, all coordinates on this path
are simple.
\end{proposition}
\begin{proof}
  Let $i$ be a coordinate associated with an edge on such a path. Let
  $x^i$ be a terminal on the path which is the closest to $i$. Recall
  the observation from before, that any two coordinates that yield the
  same terminal cut must appear in the optimal tree adjacent to one
  another, and furthermore, their order does not matter (see Fact
  ~\ref{fact:basic-properties}). Therefore, we may assume $i$ is
  adjacent to $x^i$. Furthermore, as $x^i$ is not an endpoint, there
  exists a good coordinate $j\neq i$ adjacent to $x^i$ on this path. It
  follows that $i$ and $j$ determine two different cuts over the set
  of terminals. Then, just as in the proof of
  Lemma~\ref{clm:correctness-step-2}, $i$ is simple, where $x^i$ is the
  unique terminal on one side of its cut matching $i$'s pattern.
\end{proof}{\qed}

Proposition~\ref{pro:terminal-on-a-good-path} means that any
non-simple coordinate corresponds to an entire path in the path
decomposition (because all edges on a path with no terminals on it
determine the same cut). We deduce that (a) the number of non-simple
contracted coordinates of weight more than $q$ is at most
$8q$. Furthermore, the number of edges on any path of length no more
than $q$ with no terminal on them is at most $8q^2$. Finally, since
the base of the recursion is executed only when \splito fails to execute, the
number of edges on paths that do have a terminal on them is at most
$16q^2$. It follows that the path decomposition of $\mathcal{T}$
contains at most $16 q^2$ edges, in addition to the edges on paths of
length more than $q$ with no terminal on them (and each such path
causes the algorithm to split exactly once).

We can now upper bound the number of edges in $\mathcal{T}^{\rm
  post}$. The forest $\mathcal{T}^{\rm post}$ contains at most $2q
\leq 2q^2$ bad edges; at most $16q^2$ edges from $\mathcal{T}$; and
all the paths between the vertices we introduce by adding the
$y(i)$-vertices and connecting them to $\mathcal{T}$. Let $i$ be a
non-simple contracted coordinate which was split upon, and let $u\to
v$ be the corresponding path on $\mathcal{T}$ which was
removed. Assume $y(i)$ resides on the same side of the cut as $u$.
Clearly, $u$ and $y(i)$ identify on all coordinates on the $u\to v$
path, but they also identify on all other good coordinates: a good
coordinate $j$ appears most once in $\mathcal{T}$, so $u$, $v$ and all
terminals on at least one side of the $u-v$ cut has the same value on
their $j$th coordinate. It follows that the path between $y(i)$ and
$u$ is of length at most $q$, and the same applies to the path between
$\overline{y(i)}$ and $v$. Therefore, the number of edges on paths we
have yet to upper bound, sum to no more than $8q \cdot 2 \cdot q =
16q^2$. Thus $\mathcal{T}^{\rm post}$ contains no more than $34q^2$
edges. This completes the proof.
\end{proof}{\qed}

\end{document}